\documentclass[12pt]{article}
\usepackage{amsfonts}
\usepackage{enumerate}
\usepackage{amsmath}
\usepackage{amsthm}
\usepackage{color}
\usepackage{addfont}
\addfont{OT1}{rsfs10}{\rsfs}

\def\1{{\bf 1}}

\begin{document}

\newtheorem{thm}{Theorem}[section]
\newtheorem{lem}[thm]{Lemma}
\newtheorem{prop}[thm]{Proposition}
\newtheorem{cor}[thm]{Corollary}
\newtheorem{defi}[thm]{Definition}
\newtheorem{remark}[thm]{Remark}
\newtheorem{result}[thm]{Result}
\newtheorem{exm}[thm]{Example}
\title
{\bf New extremal Type II $\mathbb{Z}_4$-codes of length 64 by the doubling method}
\author
{Sara Ban (sban@math.uniri.hr), ORCID: 0000-0002-1837-8701 \\
[3pt]
Sanja Rukavina (sanjar@math.uniri.hr), ORCID: 0000-0003-3365-7925\\
{\it  Faculty of Mathematics}\\
{\it  University of Rijeka, 51000 Rijeka, Croatia }\\[-15pt]
}
\date{}
\maketitle

\bigskip


\vspace*{0.2cm}

\begin{abstract}
\noindent
Extremal Type II $\mathbb{Z}_4$-codes are a class of self-dual $\mathbb{Z}_4$-codes with Euclidean weights divisible by eight and the largest possible minimum Euclidean weight for a given length. A small number of such codes is known for lengths greater than or equal to $48.$
The doubling method is a method for constructing Type II $\mathbb{Z}_4$-codes from a given Type II $\mathbb{Z}_4$-code. 
Based on the doubling method, in this paper we develop a method to construct new extremal Type II $\mathbb{Z}_4$-codes starting from an extremal Type II $\mathbb{Z}_4$-code of type $4^k$ with an extremal residue code and length $48, 56$ or $64$. 
Using this method, we construct three new extremal Type II $\mathbb{Z}_4$-codes of length $64$ and type $4^{31}2^2$. Extremal Type II $\mathbb{Z}_4$-codes of length $64$ of this type were not known before. Moreover, the residue codes of the constructed extremal $\mathbb{Z}_4$-codes are new best known $[64,31]$ binary codes and the supports of the minimum weight codewords of the residue code and the torsion code of one of these  codes form self-orthogonal $1$-designs.

\end{abstract}

\section{Introduction}
A $\mathbb{Z}_4$-code of length $n$ is a $\mathbb{Z}_4$-submodule of $\mathbb{Z}_4^n$, where $\mathbb{Z}_4$ denotes the ring of integers modulo $4$. Such a code is self-dual if it is equal to its dual code. Extremal Type II $\mathbb{Z}_4$-codes are a class of self-dual $\mathbb{Z}_4$-codes with Euclidean weights divisible by eight and the largest possible minimum Euclidean weight for a given length. A construction of such codes is of particular interest for both theoretical and practical reasons.

The construction of self-dual $\mathbb{Z}_4$-codes is an active area of research and many methods for constructing self-dual $\mathbb{Z}_4$-codes are known (see, e.g.,  \cite{Had},   
\cite{bent}, \cite{BonnecazeSoleBachoc}, \cite{Chan},  \cite{ConS}, \cite{CMR},  \cite{Hnovi},  
\cite{survey}).
 In 2012, Chan introduced a method for constructing Type II $\mathbb{Z}_4$-codes of type
$4^{k_1}2^{k_2}$ from a given Type II $\mathbb{Z}_4$-code of type $4^{k_1+1}2^{k_2-2}$ called the {\it doubling method} (see \cite{Chan}). Based on this method, new methods for constructing extremal Type II $\mathbb{Z}_4$-codes of lengths $24$, $32$ and $40$ have been developed and many new extremal Type II $\mathbb{Z}_4$-codes of lengths $32$ and $40$ have been constructed (see \cite{Had}, \cite{40doubling}, \cite{Chan}). Moreover, it is shown in \cite{40doubling} that if $\widetilde{C}$ is a $\mathbb{Z}_4$-code obtained from a Type II $\mathbb{Z}_4$-code $C$ by the doubling method, then the minimum weight of the residue code of $\widetilde{C}$ is greater than or equal to the minimum weight of the residue code of $C$. Therefore, good binary codes can be obtained as residue codes of Type II $\mathbb{Z}_4$-codes constructed by the doubling method.

In this paper we observe the application of the doubling method in a construction of extremal Type II $\mathbb{Z}_4$-codes for the next case, i.e., for lengths $48, 56$ and $64.$ A small number of extremal Type II $\mathbb{Z}_4$-codes for these lengths is known and all of them are of type $4^{\frac{n}{2}}$ or $4^{\frac{n}{4}}2^{\frac{n}{2}}$, where $n$ is the length of a code. There are two known inequivalent extremal Type II $\mathbb{Z}_4$-codes of length $48$ (\cite{BonnecazeSoleBachoc}, \cite{HaradaKitazume}) and  both codes are of type $4^{24}.$ Up to equivalence, there are six extremal Type II $\mathbb{Z}_4$-codes of length $56$, one of which is of type $4^{14}2^{28}$ (\cite {5664}) and the remaining five codes are of type $4^{28}$ (\cite{Hnovi}, \cite{HaradaLee}).
Finally, up to equivalence, one extremal Type II $\mathbb{Z}_4$-code of length $64$ and type $4^{16}2^{32}$ is known (\cite{5664}), and  recently five inequivalent extremal Type II $\mathbb{Z}_4$-codes of length $64$ and type $4^{32}$ have been constructed in \cite{Hnovi}.

Based on the doubling method for constructing Type II $\mathbb{Z}_4$-codes, we developed a method for constructing new extremal Type II $\mathbb{Z}_4$-codes starting from an extremal Type II $\mathbb{Z}_4$-code of type $4^k$ with an extremal residue code and length $48, 56$ or $64$. We constructed new extremal Type II $\mathbb{Z}_4$-codes of length $64$ and type $4^{31}2^2$. As mentioned earlier, such codes were not known before. Moreover, their residue codes are $[64,31,12]$ binary codes, which are not equivalent to the previously known codes. According to \cite{codetables}, codes with these parameters are the best known $[64,31]$ binary codes, since the codes that reach a theoretical upper bound on the minimum weight are not known.

The paper is organized as follows.  Section \ref{pre} gives definitions and basic properties of binary linear codes, codes over $\mathbb{Z}_4$ and lattices, that will be needed in our work.
In Section \ref{constr}, we present the method to construct new extremal Type II $\mathbb{Z}_4$-codes starting from an extremal Type II $\mathbb{Z}_4$-code of type $4^k$ with an extremal residue code and length $48, 56$ or $64$.
Finally, in the last section, we present computational results. We give the generator matrices for three new  extremal Type II $\mathbb{Z}_4$-codes of length $64$ as well as the weight distributions for the corresponding residue codes. For one of the constructed extremal $\mathbb{Z}_4$-codes, we obtained self-orthogonal $1$-designs from minimum weight codewords of the residue code and the torsion code.

In this work, we have used computer algebra systems GAP \cite{cite_key1Gap} and Magma \cite{magma}.

\section  {Preliminaries} \label{pre}

We assume that the reader is familiar with the basic facts of coding theory. We refer the reader to \cite{FEC} in relation to terms not defined in this paper.\\

Let $\mathbb{F}_q$ be the field of order $q$, where $q$ is a prime power.
A {\it code} $C$ over $\mathbb{F}_q$ of length $n$ is any subset of $\mathbb{F}_q^n$. A $k$-dimensional subspace of $\mathbb{F}_q^n$ is called an $[n,k]$ {\it $q$-ary linear code}. 
An element of a code is called a {\it codeword}. A {\it generator matrix} for an $[n, k]$ code $C$ is any
$k \times n$ matrix whose rows form a basis for $C.$

If $q=2$, then the code is called {\it binary}.
The {\it (Hamming) weight} of a codeword $x\in \mathbb{F}_2^n$ is the number of nonzero coordinates in $x.$ If the minimum weight $d$ of an $[n,k]$ binary linear code is known, then the code is called a $[n,k,d]$ binary linear code.
Binary linear codes in which all codewords have weights divisible by four are called {\it doubly even}.
 Denote by $W_i$ the number of codewords with weight $i$ in a binary code of length $n$. The {\it weight distribution} of a binary code $C$ is the set
$$\{\left\langle i,W_i\right\rangle \ |\ i=0,\dots, n, W_i\neq 0\}.$$

Let $C$ be a binary linear code of length $n$. The {\it dual code} $C^\bot$ of $C$ is defined as
$$C^\bot=\{x\in \mathbb{F}_2^n\ |\ \left\langle x, y\right\rangle=0\ \text{for all}\ y\in C\},$$
where $\left\langle x,\\ y\right\rangle=x_1y_1+x_2y_2+\dots+x_ny_n$ for 
$x=(x_1,x_2,\dots,x_n)$ and $y=(y_1,y_2,\dots,y_n).$ The code $C$ is {\it self-orthogonal} if $C\subseteq C^\bot$, and it is {\it self-dual} if $C = C^\bot$.

A self-dual doubly even binary code is called a {\it Type II binary code}.
If $C$ is a Type II binary code of length $n$ and minimum weight $d,$ then $d\leq 4\left\lfloor \frac{n}{24}\right\rfloor+4$ (see \cite{macwilliams}).
If $C$ is a Type II binary code of length $n$ and  minimum weight  $d=4\left\lfloor \frac{n}{24}\right\rfloor+4,$ we say that $C$ is {\it extremal}.

Let $\mathbb{Z}_4$ be the ring of integers modulo $4.$ A {\it linear code} $C$ of length $n$ over $\mathbb{Z}_4$ (i.e., a {\it $\mathbb{Z}_4$-code}) is a $\mathbb{Z}_4$-submodule of $\mathbb{Z}_4^n.$
Two $\mathbb{Z}_4$-codes are {\it (monomially) equivalent} if one can be obtained from the other by permuting the coordinates and (if necessary) changing the signs of certain coordinates. Codes that differ  only by a permutation of coordinates are called {\it permutation equivalent}.
The {\it support} of a codeword $x\in \mathbb{Z}_4^n$ is the set of nonzero positions in $x.$
Denote the number of coordinates $i$ (where $i=0,1,2,3$) in a codeword $x\in \mathbb{Z}_4^n$  by $n_i(x).$
The codeword $x\in \mathbb{Z}_4^n$ is {\it even} if $n_1(x)=n_3(x)=0$.
The {\it Euclidean weight} of $x$ is $wt_E (x) = n_1(x) + 4n_2(x) + n_3(x).$  
We denote by $d_E(C)$ the minimum Euclidean weight of the code $C.$  

Let $C$ be a $\mathbb{Z}_4$-code of length $n$. The {\it dual code} $C^\bot$ of the code $C$ is defined as
\[
C^\bot=\{x\in \mathbb{Z}_4^n\,|\,\left\langle x, y\right\rangle=0\ \text{for all}\ y\in C\},
\]
where $\left\langle x,\\ y\right\rangle=x_1y_1+x_2y_2+\dots+x_ny_n\ (\text{mod}\ 4)$ for
$x=(x_1,x_2,\dots,x_n)$ and $y=(y_1,y_2,\dots,y_n).$ The code $C$ is {\it self-orthogonal} if $C\subseteq C^\bot$ and {\it self-dual} if $C= C^\bot.$

{\it Type II $\mathbb{Z}_4$-codes} are self-dual $\mathbb{Z}_4$-codes which have the property that all 
Euclidean weights are divisible by eight. 
If $C$ is a Type II $\mathbb{Z}_4$-code of length $n,$ then  $n\equiv 0\ (\text{mod}\ 8)$ and $d_E(C)\leq 8\left\lfloor \frac{n}{24}\right\rfloor+8$ (see \cite{survey}). A Type II $\mathbb{Z}_4$-code $C$  length $n$ is {\it extremal} if $d_E(C)=8\left\lfloor \frac{n}{24}\right\rfloor+8$.\\

Each $\mathbb{Z}_4$-code $C$ contains a set of $k_1+k_2$ codewords $\{c_1,c_2,\dots,c_{k_1},c_{k_1+1}, \dots,$ $c_{k_1+k_2}\}$ such that every codeword in $C$ is uniquely expressible in the form
\[
\sum_{i=1}^{k_1}a_ic_i+\sum_{i=k_1+1}^{k_1+k_2}a_ic_i,
\]
where $a_i\in \mathbb{Z}_4$ and $c_i$ has at least one coordinate equal to 1 or 3, for $1\leq i\leq k_1,$
$a_i\in \mathbb{Z}_2$  and $c_i$ is even, for $k_1+1\leq i\leq k_1+k_2.$
We say that $C$ is of {\it type} $4^{k_1}2^{k_2}.$
The matrix whose rows are  $c_i,\ 1\leq i\leq k_1+k_2,$ is called a {\it generator matrix} for $C.$ A generator matrix $G$ of a $\mathbb{Z}_4$-code $C$ is in {\it standard form} if

\begin{equation*}\label{std}
G=\left[\begin{tabular}{ccc}
$I_{k_1}$&$A$&$B_1+2B_2$\\
$O$&$2I_{k_2}$&$2D$\\
\end{tabular}\right],
\end{equation*}
where $A, B_1,B_2$ and $D$ are matrices with entries from $\{0,1\}$, $O$ is the $k_2\times k_1$ null matrix,
and $I_{m}$ denotes the identity matrix of order $m$.
Each $\mathbb{Z}_4$-code is permutation equivalent to a code with generator matrix in standard form.

Let $C$ be a $\mathbb{Z}_4$-code of length $n.$  
There are two binary linear codes of length $n$ associated with $C$, namely the binary  code
$C^{(1)}=\{c\ (\text{mod}\ 2)\,|\,c\in C \}$, which is called the {\it residue code} of $C$, and the binary code $C^{(2)}=\{c\in \mathbb{Z}_2^n\,|\,2c\in C \}$, 
which is called the {\it torsion code} of $C$. 
If $C$ is a $\mathbb{Z}_4$-code of type $4^{k_1}2^{k_2}$ with a generator matrix in standard form, then $C^{(1)}$ is a binary code of dimension $k_1$ generated by the matrix
\[
\left[\begin{tabular}{ccc}
$I_{k_1}$&$A$&$B_1$\\
\end{tabular}\right].
\]
If $C$ is a self-dual $\mathbb{Z}_4$-code, then $C^{(1)}$ is doubly even and $C^{(1)}={C^{(2)}}^\bot$ (see \cite{ConS}).

Let $\{v_1,\dots,v_n\}$ be a basis of $\mathbb{R}^n.$ A {\it lattice} $\Lambda$ is a set
$$\Lambda=\{z_1v_1+\dots+z_nv_n \ |\ z_i\in \mathbb{Z}, 1\leq i\leq n\}.$$

Let $\Lambda\subseteq \mathbb{R}^n$ be a lattice.
If $\left\langle x,y\right\rangle\in \mathbb{Z},$ for all $x,y\in \Lambda,$ where
$\left\langle x,y\right\rangle=x_1y_1+\dots +x_ny_n,\ x=(x_1\dots,x_n),\ y=(y_1,\dots,y_n),$ $\Lambda$ is an {\it integral} lattice.
An integral lattice $\Lambda$ is {\it even} if $\left\langle x,x\right\rangle$ is an even number for all $x\in \Lambda.$

If $\Lambda\subseteq \mathbb{R}^n$ is a lattice, then its {\it dual lattice} is 
$$\Lambda^*=\{y\in \mathbb{R}^n  \ |\ \left\langle x,y\right\rangle\in \mathbb{Z}\ \text{for all}\ x\in \Lambda \}.$$ 
An integral lattice $\Lambda$ is {\it unimodular} if $\Lambda=\Lambda^*.$

 The {\it minimum norm} $\mu$ of a lattice $\Lambda$ is
$$\mu=\min \{ \left\langle x,x\right\rangle \ |\ x\in \Lambda, x\neq (0,\dots,0)\}.$$
The {\it kissing number} of $\Lambda$ is the number of elements of $\Lambda$ with minimum norm $\mu.$

For a $\mathbb{Z}_4$-code $C$ of length $n,$ define the following lattice
$$A_4 (C)=\{x\in \mathbb{R}^n  \ |\ 2x\ \text{(mod 4)} \in C \}.$$

$A_4 (C)$ is an even unimodular lattice if and only if $C$ is a Type II $\mathbb{Z}_4$-code (\cite{FEC}, p. 504).\\

To determine the extremality of Type II $\mathbb{Z}_4$-codes of lengths $48, 56$ and $64$, one can use the following lemma. This result provides a faster computational option than computing the Euclidean weight distribution in Magma \cite{magma}.

\begin{lem}[\cite{Hnovi}, Lemma 6.1.] \label{res}
Suppose that $n\in \{48, 56, 64\}.$ Let $C$ be a Type II $\mathbb{Z}_4$-code of
length $n.$ Then
$C$ is extremal if and only if $A_4(C)$ has minimum norm $4$ and kissing number
$2n.$
\end{lem}

\section{Method of construction} \label{constr}

In \cite{Chan}, the doubling method for a construction of Type II $\mathbb{Z}_4$-codes is introduced.

\begin{thm}[\cite{Chan}, Doubling method]\label{TmDoubling}
Let $C$ be a Type II $\mathbb{Z}_4$-code of length $n.$ 
Let $2u\in \mathbb{Z}_4^{n}$ be an even codeword. Suppose $2u \notin C$ and $2u$ has an even number of $2'$s in its coordinates.
Let $C_0=\{v\in C\ |\ \left\langle 2u,v\right\rangle=0\}.$
Then $\widetilde{C}=C_0 \oplus \left\langle 2u\right\rangle$ is a Type II $\mathbb{Z}_4$-code.
\end{thm}

Let $C$ be a Type II $\mathbb{Z}_4$-code of length $n$ and type $4^k.$ Then
the choice of $2u$ in the previous theorem can be restricted to codewords with $0'$s on the first $k$ coordinates (see [\cite{Chan}, Theorem 4]).\\

In the sequel, we consider a  construction of extremal Type II $\mathbb{Z}_4$-codes of length $n \in \{48,56,64\}$ and type $4^k$ which have an extremal residue code, using the doubling method.

\begin{thm}\label{doubling4k485664}
Let $n \in \{48,56,64\}$. Denote by $S_i(w)$ the set of positions with element $i\in\mathbb{Z}_4$ in $w\in \mathbb{Z}_4^n.$ Let $C$ be an extremal Type II $\mathbb{Z}_4$-code of length $n$ and type $4^{k}$ where $C^{(1)}$ is extremal.
Suppose $2u\in\mathbb{Z}_4^{n}$ such that $S_2(2u)\subseteq \{k+1,\dots,n\},$ where $|S_2(2u)|\geq 6$ is even. 
If there is no codeword $v$ of $C$ that satisfies any of the following conditions:
\begin{enumerate}
\item $S_2(v)\subseteq S_2(2u)\subseteq S_1(v)\cup S_2(v) \cup S_3(v)$ and $|S_1(v)\cup S_3(v)|=16,$
\item $|S_2(2u)\setminus (S_1(v)\cup S_2(v)\cup S_3(v))|+|S_2(v)\setminus S_2(2u)|=1$ and $|S_1(v)\cup S_3(v)|=12,$
\item $|S_2(2u)\setminus S_2(v)|+|S_2(v)\setminus S_2(2u)|=2$ or $4$ and $|S_1(v)\cup S_3(v)|=0,$
\end{enumerate}
then the Type II $\mathbb{Z}_4$-code $\widetilde{C}$ generated by $2u$ and $C$ using the doubling method is extremal. These choices of $2u$ are the only candidates for the code $C$ in the doubling method which lead to an extremal code.
\end{thm}

\begin{proof}  

Let us assume that the code $\widetilde{C}$ is not extremal. Then it contains a codeword of Euclidean weight $8$ or $16$ of the form  $w=v+2u,$ where $v\in C$ is such a codeword that $\left\langle v,2u\right\rangle=0,$ and 
$$wt_E(w)=\left|S_1(v)\cup S_3(v)\right|+4\left(\left|S_2(v)\cap S_0(2u)\right|+\left|S_0(v)\cap S_2(2u)\right|\right).$$
There are three cases to consider.\\
Case 1: $\left|S_1(v)\cup S_3(v)\right|>16.$\\
Then $wt_E(w)\geq 24.$\\
Case 2: $\left|S_1(v)\cup S_3(v)\right|=16.$\\
For $wt_E(w)$ to be equal to $16,$ $S_2(v)\cap S_0(2u)$ and $ S_0(v)\cap S_2(2u)$ must be empty sets, which is impossible because of the first condition.\\
Case 3: $\left|S_1(v)\cup S_3(v)\right|<16.$\\
If $\left|S_1(v)\cup S_3(v)\right|=12,$ then, for $wt_E(w)$ to be $16,$ we have $\left(\left| S_2(v)\cap S_0(2u)\right|,\left| S_0(v)\cap S_2(2u)\right|\right)\in \{(0,1),(1,0)\},$ which is impossible because of the second condition.\\
If $v$ is an even codeword, then $wt_E(w)=4\left(\left|S_2(v)\cap S_0(2u)\right|+\left|S_0(v)\cap S_2(2u)\right|\right).$ For $wt_E(w)$ to be $8$ or $16,$ we have $|S_2(2u)\setminus S_2(v)|+|S_2(v)\setminus S_2(2u)|=2$ or $4,$ which is impossible because of the third condition.

The resulting choices for $2u$  are the only candidates for the code $C$ in the doubling method, since the conditions of the theorem exclude all choices that lead to a code $\widetilde{C}$ which is not extremal.
\end{proof}

For an extremal Type II $\mathbb{Z}_4$-code $C,$ the next algorithm returns all unsuitable candidates $2u$,  i.e., the candidates for which  the application of the doubling method leads to a Type II $\mathbb{Z}_4$-code $\widetilde{C}$ which is not extremal. Thus, performing the given steps will find all possible candidates $2u$ for code $C$ to produce a new extremal Type II $\mathbb{Z}_4$-code $\widetilde{C}$ by the doubling method.\\

 {\bf Algorithm B} \\\label{doubling4k485664alg}
Let $n \in \{48, 56, 64\}$  and let $C$ be an extremal Type II $\mathbb{Z}_4$-code of length $n$ and type $4^{k},$ where $C^{(1)}$ is extremal, with the generator matrix 
$G=\left[\begin{tabular}{cc}
$I_{k}$&$A$
\end{tabular}\right]$ in the standard form.

\begin{itemize}
\item[1.] Let $G^{(1)}=\left[I_{k}\ \ A^1\right]$ be a generator matrix of $C^{(1)}.$
\item[2.] For each $F\subseteq \{1,2,\dots,k\}$ such that $1\leq |F|\leq 16$ we do the following.
\item[2.1.] Find the sum $s^1_F$ of rows in $G^{(1)}$ with row indices in $F$. 
\item[2.2.] If $|S_1(s^1_F)|=16,$  we repeat the following steps on all $E\subseteq F.$
\begin{itemize}
\item[2.2.1.] Calculate $v=s_F+2s_E,$ where $s_F$ is the sum of rows in the generator matrix $G$ of $C$ with row indices in $F$ and $s_E$ is the sum of rows in $G$ with row indices in $E.$
\item[2.2.2.] Let $$T_v=S_2(v)\cap \{k+1,k+2,\dots,n\},$$ $$O_v=(S_1(v)\cup S_3(v))\cap \{k+1,k+2,\dots,n\}.$$
\item[2.2.3.] Let $\mathcal{B}$ be the collection of all sets  $$B=T_v\cup O,\ O\subseteq O_v,$$ where $|B|$ is an even number.
 \end{itemize}
\item[2.3.] If $|S_1(s_F^1)|=12,$ we repeat the following steps on all $E\subseteq F.$
\begin{itemize}
\item[2.3.1.] Evaluate $v=s_F+2s_E$ where $s_F$ is the sum of rows in $G$ with row indices in $F$ and $s_E$ is the sum of rows in $G$ with row indices in $E.$
\item[2.3.2.] Let $T_v=S_2(v)\cap \{k+1,\dots,n\}$ and $O_v=\left(S_1(v)\cup S_3(v)\right)\cap \{k+1,\dots,n\}.$
\item[2.3.3.] Consider all $B\subseteq \{k+1,\dots,n\}$ such that $T_v\subseteq B$ and $|B\setminus \left(T_v\cup O_v\right)|=1.$ Include all sets $B$ in $\mathcal{B}.$
\item[2.3.4.] Let $B\subseteq \left(T_v\cup O_v\right)$ such that $|T_v\setminus B|=1.$ Include all such sets in $\mathcal{B}.$
\item[2.3.5.] Evaluate $v_i=v+2G_i$ for all $i \in \{1,\dots,k\}\setminus F.$
\item[2.3.6.] Let $T_{v_i}=S_2(v_i)\cap \{k+1,\dots,n\}$ and $O_{v_i}=\left(S_1(v_i)\cup S_3(v_i)\right)\cap \{k+1,\dots,n\}.$
\item[2.3.7] Include all sets $T_{v_i}\cup O, O\subseteq O_{v_i}$ in $\mathcal{B}.$
\end{itemize}
\item[3.] For all $i\in \{1,\dots,k\},$ we do the following.
\begin{itemize}
\item[3.1.] Let $O_i=S_2(2G_i)\cap \{k+1,\dots,n\},$ where $G_i$ is the $i$-th row of $G.$
\item[3.2.] Include all  $B\subseteq O_i$ such that $|O_i|-|B|=1$ or $|O_i|-|B|=3$ in $\mathcal{B}.$
\item[3.3.] Include all $B=O_i\cup \{m\}$ such that $m\in \{k+1,\dots,n\}\setminus O_i$ in $\mathcal{B}.$
\item[3.4.] Include all sets $B=O_i\cup \{p,q,r\}$ for every $\{p,q,r\}\subseteq \{k+1,\dots,n\}\setminus O_i$ in $\mathcal{B}.$
\end{itemize}
\item[4.] For every $\{i,j\}\subseteq \{1,\dots,k\},$ we do the following.
\begin{itemize}
\item[4.1.] Let $v_{i,j}=2G_i+2G_j$ and $O_{i,j}=S_2(v_{i,j})\cap \{k+1,\dots,n\}.$ 
\item[4.2.] Include all $O_{i,j}$ in $\mathcal{B}.$
\item[4.3.] Include all $B=O_{i,j}\cup \{p,q\}$ for every $\{p,q\}\subseteq \{k+1,\dots,n\}\setminus O_{i,j}$ in $\mathcal{B}.$
\end{itemize}
\item[5.] For every $\{i,j,l\}\subseteq \{1,\dots,k\},$ we do the following.
\begin{itemize}
\item[5.1.] Let $v_{i,j,l}=2G_i+2G_j+2G_l$ and $O_{i,j,l}=S_2(v_{i,j,l})\cap \{k+1,\dots,n\}.$ 
\item[5.2.] Include all $B\subseteq O_{i,j,l}$ such that $|O_{i,j,l}|-|B|=1$ in $\mathcal{B}.$
\item[5.3.] Include all $B=O_{i,j,l}\cup \{m\}$ such that $m \in \{k+1,\dots,n\}\setminus O_{i,j,l}$ in $\mathcal{B}.$
\end{itemize}
\item[6.] For every $\{i,j,l,m\}\subseteq \{1,\dots,k\},$ we do the following.
\begin{itemize}
\item[6.1.] Let $v_{i,j,l,m}=2G_i+2G_j+2G_l+2G_m$ and $O_{i,j,l,m}=S_2(v_{i,j,l,m})\cap \{k+1,\dots,n\}.$ 
\item[6.2.] Include all $O_{i,j,l,m}$ in $\mathcal{B}.$
\end{itemize}
\end{itemize}

Our method of construction is based on the following theorem.

\begin{thm}\label{TmAlgB}
Let $n \in \{48, 56, 64\}.$ Denote by $S_i(w)$ the set of positions with element $i\in\mathbb{Z}_4$ in $w\in \mathbb{Z}_4^n.$ Let $C$ be an extremal Type II $\mathbb{Z}_4$-code of length $n$ and type $4^{k},$ where $C^{(1)}$ is extremal.
Furthermore, let $\mathcal{S}$ be the collection of all $S\subseteq\{k+1,\dots,n\}$ such that $|S|$ is even and $|S|\geq 6.$
Then $\mathcal{G}=\mathcal{S}\setminus \mathcal{B}$ is the set of all possible $S_2(2u)$ for the code $C$ in the doubling method which lead to an extremal code, where $\mathcal{B}$ is the set obtained by applying Algorithm B.
\end{thm}

\begin{proof} The first condition in Theorem \ref{doubling4k485664} is checked in Step 2.2. of Algorithm B. Since the condition requires that $S_2(v)\subseteq S_2(2u),$ the coefficients of the rows of  $G$ in the linear combination of $v$ can only be $0, 1$ or $3.$ Step 2.2.1 generates all such codewords $v$ with $|S_1(v)\cup S_3(v)|=16.$ All subsets $B=S_2(2u)$ satisfying the first condition are included in $\mathcal{B}$ in Step 2.2.3.

The second condition in Theorem \ref{doubling4k485664} is checked in Step 2.3. If $|S_2(v)\setminus S_2(2u)|=0,$  the coefficients of rows of  $G$ in the linear combination of $v$ can only be $0, 1$ or $3.$  Step 2.3.1 generates all such codewords $v$ with $|S_1(v)\cup S_3(v)|=12.$ 
When $|S_2(v)\setminus S_2(2u)|=1,$ $v$ is a linear combination of rows of $G$ with at most one coefficient $2.$ 
All codewords $v$ with exactly one coefficient $2$ satisfying $|S_1(v)\cup S_3(v)|=12$  are constructed in Step 2.3.5.
All subsets $B=S_2(2u)$ satisfying the second condition are included in $\mathcal{B}$ in Steps 2.3.3., 2.3.4 and 2.3.7.

Since $|S_2(v)\setminus S_2(2u)|\in \{1,2,3,4\}$ in the third condition of Theorem \ref{doubling4k485664}, $v$ is the sum of at most four rows of $G$ with coefficients $2.$
These codewords are considered in Steps 3., 4., 5. and 6.
 
\end{proof}

\section{Computational results} \label{results}

We consider all known extremal Type II $\mathbb{Z}_4$-codes of lengths $48$, $56$ and $64$ of type $4^k$.

\subsection{Lengths n=48 and n=56} \label{s48}

Two inequivalent Type II $\mathbb{Z}_4$-codes of length $48$ are known (\cite{BonnecazeSoleBachoc}, \cite{HaradaKitazume}) with notation $C_{48p}^{(4)}$ from \cite{HaradaKitazume} and  $\mathcal{D}_{48}$ from \cite{HaradaLee}. Both codes are of type $4^{24}$.

The extremal Type II $\mathbb{Z}_4$-code  $C_{48p}^{(4)}$ of type $4^{24}$ and length $48$ has a residue code of minimum weight $8$, and it cannot be used to construct a new extremal Type II $\mathbb{Z}_4$-code by applying  Theorem \ref{doubling4k485664}.

The residue code of $\mathcal{D}_{48}$ has the minimum weight $12,$ so this code is suitable for applying Theorem \ref{doubling4k485664} to obtain new extremal codes by the doubling method. However, using Algorithm B, we found that there are no candidates $2u$ for a construction of extremal Type II $\mathbb{Z}_4$-codes of type $4^{23}2^2$ and length $48$ from $\mathcal{D}_{48}$ by the doubling method. \\

Five inequivalent extremal Type II $\mathbb{Z}_4$-codes of length $56$ and type $4^{28}$ are known: $\mathcal{C}_{56}$ and $\mathcal{D}_{56,1}$ from \cite{HaradaLee}, and $C_{4,56,i},$ for $i=1,2,3$ from \cite{Hnovi}.

 The extremal Type II $\mathbb{Z}_4$-codes $C_{4,56,1}$ and $C_{4,56,2}$  of type $4^{28}$ and length $56$ have residue codes of minimum weight $8.$ So, we cannot use Theorem \ref{doubling4k485664} to obtain new extremal Type II $\mathbb{Z}_4$-codes from $C_{4,56,1}$ and $C_{4,56,2}$ using the doubling method. 

The minimum weights of the residue codes of $\mathcal{C}_{56}, \mathcal{D}_{56,1}$ and $C_{4,56,3}$   are equal to $12$ and therefore the codes $\mathcal{C}_{56}, \mathcal{D}_{56,1}$ and $C_{4,56,3}$ are suitable for an application of Theorem \ref{doubling4k485664}.  We applied Algorithm B and found that there are no candidates $2u$ for a construction of extremal Type II $\mathbb{Z}_4$-codes of type $4^{27}2^2$ and length $56$ from $\mathcal{C}_{56}, \mathcal{D}_{56,1}$ and $C_{4,56,3}$ by doubling method.\\

\subsection{New extremal Type II $\mathbb{Z}_4$-codes of length n=64}  \label{s64}

Recently, Harada proved the existence of five inequivalent extremal Type II $\mathbb{Z}_4$-codes $C_{4,64,i},$  $i=1,\dots,5,$ of length $64$ and type $4^{32}$ (\cite{Hnovi}).

The minimum weights of the residue codes of $C_{4,64,1}, C_{4,64,2}$ and $C_{4,64,5}$  are equal to $8$, and they are not suitable for applying   Theorem \ref{doubling4k485664} to obtain new extremal Type II $\mathbb{Z}_4$-codes by the doubling method. 

The residue codes of the extremal Type II $\mathbb{Z}_4$-codes $C_{4,64,3}$  and $C_{4,64,4}$ of type $4^{32}$ and length $64$ have minimum weight $12.$ So, we can apply Theorem \ref{doubling4k485664}.

For length $64,$ the search in Algorithm B takes too much time. Therefore, we excluded most of the unsuitable $S_2(2u)$ for $C_{4,64,3}$ and performed a pseudo-random search on the remaining sets. We obtained three new extremal Type II $\mathbb{Z}_4$-codes of type $4^{31}2^2$ and length $64$ for $S_2(2u)$ equal to $\{ 33, 34, 35, 36, 37, 38, 39, 40, 41, 42, 43, 44, 45, 46, 47, 48, 49, 50, 51, 52, 53,\\ 54,$ $ 55, 56 \}, \{ 33, 34,\dots, 64 \}$ and $\{ 33, 34, 35, 36, 37, 38, 39, 40, 41, 42, 43, 44, 45, 46, 47, 48, 49,\\ 50, 51, 52, 53, 54, 55, 56, 57, 58, 59, 60, 61, 63\},$ which will be denoted by $\widetilde{C_{4,64,3}}_i,$ $i=1,2,3,$  respectively.

 The generator matrix of $\widetilde{C_{4,64,3}}_i,\ i=1,2,3,$ in standard form is
\begin{equation*}\label{std1}
\widetilde{G}_i=\left[\begin{tabular}{ccc}
$I_{31}$&$(A)_i$&$(B_1+2B_2)_i$\\
$O$&$2I_{2}$&$(2D)_i$\\
\end{tabular}\right],
\end{equation*}
where 
\begin{equation*}
{(A)_1}^T=\left[\begin{tabular}{c}
1 1 0 1 1 0 0 1 0 0 1 0 0 1 1 0 0 1 1 1 0 0 0 1 0 1 1 1 0 0 0\\
1 1 1 0 1 0 0 1 0 0 1 0 0 0 0 0 1 0 1 1 1 1 0 0 0 1 1 0 1 0 1
\end{tabular}\right],
\end{equation*}
\begin{equation*}
{(A)_2}^T=\left[\begin{tabular}{c}
1 1 1 1 1 1 1 1 1 1 1 1 1 1 1 1 1 1 1 1 1 1 1 1 1 1 1 1 1 1 1\\
1 1 1 0 1 0 0 1 0 0 1 0 0 0 0 0 1 0 1 1 1 1 0 0 0 1 1 0 1 0 1 
\end{tabular}\right],
\end{equation*}
\begin{equation*}
{(A)_3}^T=\left[\begin{tabular}{c}
1 1 1 1 0 0 0 1 0 0 1 1 0 1 1 1 1 1 1 0 1 0 0 1 0 0 1 1 0 0 0\\
1 1 1 0 1 0 0 1 0 0 1 0 0 0 0 0 1 0 1 1 1 1 0 0 0 1 1 0 1 1 0 
\end{tabular}\right],
\end{equation*}
\begin{equation*}
(2D)_1=\left[\begin{tabular}{c}
2 0 2 0 2 2 0 0 0 2 2 2 2 0 2 2 2 0 2 0 0 2 0 0 2 0 0 0 0 0 2\\
2 2 2 2 2 2 2 2 2 2 2 2 2 2 2 2 2 2 2 2 2 2 2 0 0 0 0 0 0 0 0
\end{tabular}\right],
\end{equation*}
\begin{equation*}
(2D)_2=\left[\begin{tabular}{c}
2 0 2 0 2 2 0 0 0 2 2 2 2 0 2 2 2 0 2 0 0 2 0 0 2 0 0 0 0 0 2\\
 2 2 2 2 2 2 2 2 2 2 2 2 2 2 2 2 2 2 2 2 2 2 2 2 2 2 2 2 2 2 2
\end{tabular}\right],
\end{equation*}
\begin{equation*}
(2D)_3=\left[\begin{tabular}{c}
2 0 2 2 0 0 0 2 2 2 2 0 2 0 2 0 2 0 0 2 0 0 2 0 0 0 0 0 2 2 2\\
 2 2 2 2 2 2 2 2 2 2 2 2 2 2 2 2 2 2 2 2 2 2 2 2 2 2 2 2 0 2 0
\end{tabular}\right],
\end{equation*}
{\scriptsize\begin{equation*}
(B_1+2B_2)_1=\left[\begin{tabular}{c}
  3 0 3 0 1 2 0 2 1 3 1 0 3 3 0 2 3 1 2 3 1 1 2 0 0 3 0 1 0 1 1 \\
   2 0 3 2 3 1 1 0 0 0 1 3 2 2 0 3 2 2 0 1 3 2 0 0 1 3 1 2 1 0 0 \\
   1 1 2 2 2 0 2 3 2 0 1 0 0 1 2 1 0 3 0 3 1 3 1 2 0 0 3 3 2 1 0 \\  0 1 0 0 3 3 0 2 1 3 3 2 3 2 0 3 2 0 2 2 1 0
      1 3 3 0 2 3 3 2 0 \\  1 1 0 3 3 3 2 0 0 0 1 1 0 2 3 2 3 1 3 1 2 2 3 1 0 2 2 0 3 3 1 \\
   3 2 1 1 1 2 2 2 0 2 3 2 0 1 0 0 1 2 1 0 3 0 3 3 1 3 2 0 0 3 3 \\
   0 3 2 1 1 1 2 2 2 0 2 3 2 0 1 3 0 1 2 1 0 3 0 1 3 1 3 2 0 0 3 \\
   3 0 2 0 2 2 1 2 0 3 3 3 2 0 3 0 0 0 0 0 3 3 1 2 2 3 3 3 2 0 3 \\  3 0 0 3 2 1 1 1 2 2 2 0 2 3 2 2 3 3 0 1 2 1
      0 1 2 1 3 1 3 2 0 \\  2 1 2 2 1 0 3 3 3 0 0 0 2 0 1 0 0 1 1 2 3 0 3 2 1 2 1 3 1 3 2 \\
   3 2 0 0 3 2 0 3 1 0 3 1 3 0 3 3 1 0 0 3 0 2 2 3 3 1 0 1 3 1 2 \\
   3 0 2 1 2 2 1 0 3 3 3 0 0 0 2 1 2 0 0 1 1 2 3 0 3 2 1 2 1 3 1 \\
   2 1 2 0 3 0 0 3 2 1 1 1 2 2 2 1 3 0 2 2 3 3 0 3 0 3 2 1 2 1 3 \\  3 2 0 0 1 0 0 0 1 3 0 2 0 0 1 0 2 3 3 0 0 2
      1 2 0 0 1 2 1 2 0 \\  1 0 1 3 3 1 3 0 2 0 1 2 0 3 1 2 0 1 2 2 0 1 1 1 3 3 2 3 2 1 1 \\
   0 0 2 0 3 0 2 1 2 2 1 0 3 3 3 2 3 1 3 1 2 0 0 1 1 2 3 0 3 2 1 \\
   0 3 2 3 0 1 1 2 2 0 1 3 1 3 0 1 1 1 2 3 0 0 3 2 0 3 0 2 0 0 0 \\
   2 2 0 2 2 3 3 3 2 1 1 0 0 1 0 1 0 3 2 2 3 1 0 1 3 0 1 0 2 0 3 \\  1 3 1 3 1 2 2 3 1 1 3 3 2 3 2 1 1 3 2 1 2 0
      0 2 2 2 2 3 0 2 3 \\  0 0 0 2 0 3 3 0 3 2 1 3 3 3 2 3 3 2 0 3 3 1 1 2 3 1 0 0 3 0 1 \\
   3 1 0 1 0 3 2 3 0 1 1 2 2 0 1 0 2 2 2 1 1 1 2 1 2 2 1 2 0 3 0 \\
   3 1 3 2 3 2 1 0 1 2 3 3 0 0 2 0 2 0 0 0 3 3 3 0 1 2 2 1 2 0 3 \\
   3 1 3 1 0 1 0 3 2 3 0 1 1 2 2 3 2 0 2 2 2 1 1 3 0 1 2 2 1 2 0 \\  3 3 0 1 2 1 1 0 1 3 2 1 0 3 1 3 0 2 3 0 0 1
      3 3 0 0 3 2 2 1 1 \\  0 2 1 3 1 3 2 3 2 1 0 1 2 3 3 2 0 1 0 2 0 0 0 3 3 3 0 1 2 2 1 \\
   1 2 3 1 2 0 1 0 3 1 2 3 2 2 0 2 1 2 2 0 2 1 0 0 0 3 1 0 1 2 1 \\
   2 2 1 2 0 2 3 1 2 2 3 0 1 1 3 1 2 0 1 1 0 3 0 0 1 3 1 3 0 1 1 \\
   0 1 3 2 3 2 3 1 1 3 2 3 0 2 0 3 3 3 1 2 3 3 0 0 1 0 1 3 3 0 0 \\  0 3 3 2 2 0 3 1 3 1 0 1 0 3 2 1 0 0 1 0 2 3
      2 2 0 0 0 3 3 3 0 \\  3 2 1 1 0 0 2 1 3 1 3 2 3 2 1 0 3 2 2 3 2 0 1 0 2 0 0 0 3 3 3 \\
   2 1 0 3 3 2 2 0 3 1 3 1 0 1 0 3 2 1 0 0 1 0 2 1 0 2 0 0 0 3 3 
 \end{tabular}\right],
\end{equation*}}
{\scriptsize\begin{equation*}
(B_1+2B_2)_2=\left[\begin{tabular}{c}
 3 0 3 0 1 2 0 2 1 3 1 0 3 3 0 2 3 1 2 3 1 1 2 0 0 3 0 1 0 1 1 \\
   2 0 3 2 3 1 1 0 0 0 1 3 2 2 0 3 2 2 0 1 3 2 0 0 1 3 1 2 1 0 0 \\
   2 3 3 2 1 3 0 1 2 3 2 3 1 1 3 0 3 1 1 3 1 0 1 2 1 0 1 3 2 1 3 \\ 
 0 1 0 0 3 3 0 2 1 3 3 2 3 2 0 3 2 0 2 2 1 0 1 1 1 2 0 1 1 0 2 \\ 
 1 1 0 3 3 3 2 0 0 0 1 1 0 2 3 2 3 1 3 1 2 2 3 3 2 0 0 2 1 1 3 \\
   0 0 2 1 0 1 0 0 0 1 0 1 1 1 1 3 0 0 2 0 3 1 3 3 2 3 0 0 0 3 2 \\
   1 1 3 1 0 0 0 0 2 3 3 2 3 0 2 2 3 3 3 1 0 0 0 1 0 1 1 2 0 0 2 \\
   3 0 2 0 2 2 1 2 0 3 3 3 2 0 3 0 0 0 0 0 3 3 1 0 0 1 1 1 0 2 1 \\ 
 0 2 1 3 1 0 3 3 2 1 3 3 3 3 3 1 2 1 1 1 2 2 0 1 3 1 1 1 3 2 3 \\
  3 3 3 2 0 3 1 1 3 3 1 3 3 0 2 3 3 3 2 2 3 1 3 0 0 0 1 1 3 1 3 \\
   3 2 0 0 3 2 0 3 1 0 3 1 3 0 3 3 1 0 0 3 0 2 2 3 3 1 0 1 3 1 2 \\
   0 2 3 1 1 1 3 2 3 2 0 3 1 0 3 0 1 2 1 1 1 3 3 2 2 0 1 0 3 1 2 \\
   3 3 3 0 2 3 2 1 2 0 2 0 3 2 3 0 2 2 3 2 3 0 0 3 1 3 0 1 2 1 2 \\ 
 3 2 0 0 1 0 0 0 1 3 0 2 0 0 1 0 2 3 3 0 0 2 1 0 2 2 3 0 3 0 2 \\ 
 1 0 1 3 3 1 3 0 2 0 1 2 0 3 1 2 0 1 2 2 0 1 1 3 1 1 0 1 0 3 3 \\
   1 2 3 0 2 3 0 3 2 1 2 3 0 3 0 1 2 3 0 1 2 1 0 3 0 0 3 2 1 0 2 \\
   1 1 3 3 3 0 3 0 2 3 2 2 2 3 1 0 0 3 3 3 0 1 3 2 1 3 2 2 0 0 3 \\
   2 2 0 2 2 3 3 3 2 1 1 0 0 1 0 1 0 3 2 2 3 1 0 1 3 0 1 0 2 0 3 \\ 
 1 3 1 3 1 2 2 3 1 1 3 3 2 3 2 1 1 3 2 1 2 0 0 2 2 2 2 3 0 2 3 \\ 
 0 0 0 2 0 3 3 0 3 2 1 3 3 3 2 3 3 2 0 3 3 1 1 0 1 3 2 2 1 2 3 \\
   0 3 1 1 3 2 0 1 0 0 2 1 3 0 2 3 1 0 3 1 1 2 2 1 3 2 3 2 0 3 3 \\
   0 3 0 2 2 1 3 2 1 1 0 2 1 0 3 3 1 2 1 0 3 0 3 2 0 0 2 3 0 2 0 \\
   0 3 0 1 3 0 2 1 2 2 1 0 2 2 3 2 1 2 3 2 2 2 1 3 1 1 0 2 1 2 3 \\ 
 3 3 0 1 2 1 1 0 1 3 2 1 0 3 1 3 0 2 3 0 0 1 3 1 2 2 1 0 0 3 3 \\ 
 1 0 2 3 0 2 0 1 2 0 1 0 3 3 0 1 3 3 1 2 0 1 0 1 2 1 0 3 0 0 2 \\
   1 2 3 1 2 0 1 0 3 1 2 3 2 2 0 2 1 2 2 0 2 1 0 2 2 1 3 2 3 0 3 \\
   2 2 1 2 0 2 3 1 2 2 3 0 1 1 3 1 2 0 1 1 0 3 0 2 3 1 3 1 2 3 3 \\
   0 1 3 2 3 2 3 1 1 3 2 3 0 2 0 3 3 3 1 2 3 3 0 0 1 0 1 3 3 0 0 \\ 
 1 1 0 2 1 3 1 3 3 0 1 0 1 3 3 0 3 2 2 0 2 0 2 2 1 0 2 3 3 3 3 \\ 
 0 0 2 1 3 3 0 3 3 0 0 1 0 2 2 3 2 0 3 3 2 1 1 2 1 2 0 2 1 1 0 \\
   3 3 1 3 2 1 0 2 3 0 0 0 1 1 1 2 1 3 1 0 1 1 2 1 1 2 2 0 0 3 2 \\
\end{tabular}\right],
\end{equation*}}
{\scriptsize\begin{equation*}
(B_1+2B_2)_3=\left[\begin{tabular}{c}
1  0  3  1  2  3  0  1  0  1  3  3  1  1  0  3  3  1  3  2  3  0  3  2  3  1  0  3  3  2  1 \\
    0  0  3  3  0  2  1  3  3  2  3  2  0  0  0  0  2  2  1  0  1  1  1  2  0  1  1  0  0  1  0 \\
    0  3  3  3  2  0  0  0  1  1  0  2  3  3  3  1  3  1  2  2  3  3  2  0  0  2  1  1  1  2  3 \\
    2  1  0  1  0  0  0  1  0  1  1  1  1  0  0  0  2  0  3  1  3  3  2  3  0  0  0  3  2  1  0 \\
    0  3  3  3  0  0  0  2  0  1  0  2  3  2  2  3  0  3  2  1  2  1  3  1  3  2  0  0  3  3  2 \\
    3  2  1  1  1  2  2  2  0  2  3  2  0  1  0  0  1  2  1  0  3  0  3  1  3  1  0  2  0  1  3 \\
    0  3  2  1  1  1  2  2  2  0  2  3  2  0  1  3  0  1  2  1  0  3  0  3  1  3  1  0  0  2  3 \\
    1  0  2  1  3  3  1  1  3  1  1  2  0  2  3  1  0  0  1  3  1  2  2  2  3  3  1  3  1  3  3 \\
    3  0  0  3  2  1  1  1  2  2  2  0  2  3  2  2  3  3  0  1  2  1  0  3  0  3  1  3  3  0  0 \\
    2  1  2  2  1  0  3  3  3  0  0  0  2  0  1  0  0  1  1  2  3  0  3  2  1  2  1  3  1  3  2 \\
    1  2  0  1  0  3  0  2  0  2  1  0  1  2  3  0  1  0  1  2  2  1  3  1  2  3  0  3  2  2  2 \\
    2  2  3  2  2  2  3  1  2  0  2  2  3  2  3  1  1  2  2  0  3  2  0  0  1  2  1  2  0  2  0 \\
    2  1  2  0  3  0  0  3  2  1  1  1  2  2  2  1  3  0  2  2  3  3  0  1  2  1  0  3  2  3  3 \\
    1  2  0  1  2  1  0  3  0  1  2  1  2  2  1  1  2  3  0  3  2  1  2  2  1  0  3  2  0  1  0 \\
    3  0  1  0  0  2  3  3  1  2  3  1  2  1  1  3  0  1  3  1  2  0  2  1  0  3  0  3  1  0  1 \\
    3  2  3  1  3  0  0  2  1  3  0  2  2  1  0  2  2  3  1  0  0  0  1  1  3  2  3  0  2  1  0 \\
    3  1  3  0  0  1  3  3  1  1  0  1  0  1  1  1  0  3  0  2  2  0  0  0  0  1  2  0  3  1  3 \\
    0  2  0  3  3  0  3  2  1  3  3  3  2  3  0  2  0  3  3  1  1  0  1  3  2  2  1  2  1  1  3 \\
    3  3  1  0  2  3  2  2  0  3  1  2  0  1  2  2  1  3  3  0  0  3  1  0  1  0  2  1  3  3  3 \\
    3  2  3  2  1  0  1  2  3  3  0  0  2  3  1  0  0  0  3  3  3  0  1  2  2  1  2  0  3  0  2 \\
    2  3  1  2  0  3  0  0  3  2  0  0  1  2  2  0  1  0  0  0  3  1  3  3  2  0  3  0  3  0  3 \\
    3  1  3  2  3  2  1  0  1  2  3  3  0  0  2  0  2  0  0  0  3  3  3  0  1  2  2  1  2  0  3 \\
    3  1  3  1  0  1  0  3  2  3  0  1  1  2  2  3  2  0  2  2  2  1  1  1  2  3  0  0  1  0  0 \\
    1  3  0  2  3  2  1  3  0  1  0  0  2  1  1  0  0  2  0  3  2  0  0  3  1  0  1  2  1  0  1 \\
    0  2  1  3  1  3  2  3  2  1  0  1  2  3  3  2  0  1  0  2  0  0  0  3  3  3  0  1  2  2  1 \\
    0  0  2  1  3  1  3  2  3  2  1  0  1  2  3  3  2  0  1  0  2  0  0  0  3  3  3  0  1  2  2 \\
    0  2  1  3  1  3  3  0  1  0  1  3  3  3  3  2  2  0  2  0  2  2  1  0  2  3  3  3  3  0  1 \\
    2  1  3  3  0  3  3  0  0  1  0  2  2  0  0  0  3  3  2  1  1  2  1  2  0  2  1  1  2  1  0 \\
    0  3  3  2  2  0  3  1  3  1  0  1  0  3  2  1  0  0  1  0  2  3  2  0  2  2  2  1  3  1  0 \\
    2  1  0  3  3  2  2  0  3  1  3  1  0  1  0  3  2  1  0  0  1  0  2  3  2  0  2  2  0  1  3 \\
    1  2  1  0  3  3  2  2  0  3  1  3  1  0  1  3  3  2  1  0  0  1  0  2  3  2  0  2  0  2  3 
\end{tabular}\right].
\end{equation*}
}

We used Lemma \ref{res} to check the extremality of $\widetilde{C_{4,64,3}}_1, \widetilde{C_{4,64,3}}_2$ and $\widetilde{C_{4,64,3}}_3.$ We calculated the weight distributions of the corresponding residue codes using Magma (\cite{magma}). The corresponding residue codes ${\widetilde{C_{4,64,3}}_i}^{(1)}$, $i=1,2,3$,  are binary $[64,31,12]$ codes. Their weight distributions, given in Table \ref{tablica} ($W_{64-i}=W_i$), show that the constructed extremal $\mathbb{Z}_4$-codes $\widetilde{C_{4,64,3}}_i,\ i=1,2,3,$ are not equivalent.

\begin{table}[h] 
	\centering
		\begin{tabular}{|c|c|c|c|c|c|c|c|c|}
		\hline
			&$i$& 0&12&16&20&24&28&32\\
			\hline
			 ${\widetilde{C_{4,64,3}}_1}^{(1)}$&$W_i$& 1&1552&228812&9132752&116710080&521006880&853323494\\
			\hline
			${\widetilde{C_{4,64,3}}_2}^{(1)}$&$W_i$& 1&1696&228140&9124896&116763456&520871232&853504806\\
			\hline
			${\widetilde{C_{4,64,3}}_3}^{(1)}$&$W_i$&
            1&1548&227316&9136668&116716704&520969176&853380822\\
			\hline
		\end{tabular}\caption{Weight distributions of the residue codes}\label{tablica} 
\end{table}
 
 According to \cite{codetables}, binary $[64,31,12]$ codes are the best known binary $[64,31]$ codes. Furthermore, since the best known binary $[64,31]$ code from \cite{codetables} has $W_{12}=10309,$
the residue codes of $\widetilde{C_{4,64,3}}_i,\ i=1,2,3,$  are new best known binary codes with these parameters.\\

Our results, together with [\cite{Hnovi}, Proposition 6.2.], yield the following statement.
\begin{thm}
There are at least nine inequivalent extremal Type II $\mathbb{Z}_4$-codes of length $64.$
\end{thm}

\begin{remark}
{\rm
 An incidence structure ${\mathcal D} =( {\mathcal P},{\mathcal B},{\mathcal I})$, with point set ${\mathcal P}$,
block set ${\mathcal B}$ and incidence ${\mathcal I}$ is a {\it $t$-$(v,k,\lambda)$ design}, if $|{\mathcal P}|=v$, every block
$B \in {\mathcal B}$ is incident with exactly $k$ points, and every $t$ distinct points are together incident with precisely $\lambda$ blocks.

If the condition
 $$
 |B_i \cap B_j| \equiv |B_k|\equiv 0\ (\text{mod}\ 2)
 $$
 is satisfied for all blocks $B_i, B_j$ and $B_k$ of ${\mathcal D},$ we say that ${\mathcal D}$ is a {\it self-orthogonal design.
}

The supports of the minimum weight codewords in ${\widetilde{C_{4,64,3}}_2}^{(1)}$ form a self-orthogonal $1$-$(64,12,318)$ design with $1696$ blocks and block intersection numbers $0,2,4$ and $6.$
The supports of the minimum weight codewords in ${\widetilde{C_{4,64,3}}_2}^{(2)}$ form a quasi-symmetric self-orthogonal $1$-$(64,8,3)$ design with $24$ blocks and block intersection numbers $0$ and $2.$

}
\end{remark}

\end{document}